\documentclass[journal]{IEEEtran}

\usepackage{url}
\usepackage{graphicx}
\usepackage{balance}
\usepackage{epsfig,amssymb}
\usepackage{epstopdf}
\usepackage{times}
\usepackage{mathrsfs}
\usepackage{array}
\usepackage{amsmath, latexsym, amsfonts, amssymb}
\usepackage[mathscr]{eucal}
\usepackage{array, tabularx}
\usepackage[table]{xcolor}
\usepackage{multirow}
\usepackage{epsfig}
\usepackage{cite}
\usepackage{tikz}
\usepackage{mathtools}
\usepackage{psfragx}
\usepackage{xcolor}
\usepackage{bbm}
\usepackage{lineno}
\usepackage{algorithm}
\usepackage{algorithmic}
\usepackage{lipsum}
\usepackage{caption}
\usepackage{subcaption}

\usepackage[utf8]{inputenc}
\usepackage{amsmath}
\usepackage{calc}

\usepackage{accents}

\DeclareFontFamily{U}{matha}{\hyphenchar\font45}
\DeclareFontShape{U}{matha}{m}{n}{
	<-6> matha5 <6-7> matha6 <7-8> matha7
	<8-9> matha8 <9-10> matha9
	<10-12> matha10 <12-> matha12
}{}
\DeclareSymbolFont{matha}{U}{matha}{m}{n}

\DeclareFontFamily{U}{mathx}{\hyphenchar\font45}
\DeclareFontShape{U}{mathx}{m}{n}{
	<-6> mathx5 <6-7> mathx6 <7-8> mathx7
	<8-9> mathx8 <9-10> mathx9
	<10-12> mathx10 <12-> mathx12
}{}
\DeclareSymbolFont{mathx}{U}{mathx}{m}{n}

\DeclareMathDelimiter{\vvvert} {0}{matha}{"7E}{mathx}{"17}%

\DeclarePairedDelimiterX{\normiii}[1]
{\vvvert}
{\vvvert}
{\ifblank{#1}{\:\cdot\:}{#1}}

\newcolumntype{S}{>{\centering\arraybackslash}p{5em}}

\newtheorem{lemma}{Lemma}

\newtheorem{proposition}{Proposition}

\newenvironment{proof}
{\textit{Proof:} }
{$\square$}

\IEEEoverridecommandlockouts

\begin{document}
\title{Refined Nonlinear Rectenna  Modeling and Optimal  Waveform Design for Multi-User Multi-Antenna  Wireless Power Transfer}

\author{ 
	Samith Abeywickrama,   Rui Zhang,  \IEEEmembership{Fellow, IEEE},  and  Chau Yuen,  \IEEEmembership{Fellow, IEEE} 
	
	\thanks{
	
		S. Abeywickrama is with the Department of Electrical and Computer Engineering, National University of Singapore, Singapore, and also with Singapore University of Technology and Design, Singapore (e-mail: samith@u.nus.edu).
		
		R. Zhang is with the Department of Electrical and Computer Engineering, National University of Singapore, Singapore (e-mail: elezhang@nus.edu.sg).
		
		C. Yuen is with  Singapore University of Technology
		and Design, Singapore (e-mail: yuenchau@sutd.edu.sg).
	}
}

\vspace{-1cm}

\date{}
\bibliographystyle{ieeetr}
\maketitle
\begin{abstract}
In this paper, we study the optimal waveform design   for   wireless power transfer (WPT)  from  a multi-antenna energy transmitter (ET) to multiple  single-antenna energy receivers (ERs) simultaneously  in multi-path  frequency-selective channels. First, we propose  a refined nonlinear current-voltage model of the diode in the ER rectifier, and accordingly derive new expressions for the output  direct current (DC) voltage and  corresponding  harvested power at the ER. Leveraging this new rectenna model,   we first consider   the single-ER case and study the multisine-based power waveform design based on the wireless channel to maximize the harvested power at the ER. We propose two efficient  algorithms for finding high-quality suboptimal solutions to this non-convex optimization problem. Next, we extend our formulated waveform design problem to the general multi-ER case for maximizing the weighted sum of the harvested powers by all ERs, and propose an efficient difference-of-convex-functions programming (DCP)-based algorithm for solving this problem.  Finally, we demonstrate the superior  performance of our proposed waveform designs based on the new rectenna model over existing schemes/models via  simulations.
\end{abstract}


\section{Introduction}

Radio frequency (RF) transmission  enabled wireless power transfer (WPT), also known as far-field WPT, has become a promising technology for enabling reliable as well as  perpetual power supply to freely located low-power Internet of Things (IoT) devices \cite{8476597,888}. In order to maximize the WPT efficiency from  an energy transmitter (ET) to one or more energy receivers (ERs), the design of adaptive power signal waveform at the ET based on its wireless channels with all ERs is crucial \cite{8476597,888}.   

The earlier  works in WPT and WPT-enabled communications have usually assumed  the linear energy harvesting (EH) model, where the RF-to-direct current (DC) power conversion efficiency, $\eta$, at the ER is assumed to be constant regardless of its incident RF signal power and waveform. Though providing an upper bound on the energy transmission efficiency from the ET to ERs (with $\eta$ set equal to its maximum value of one at all ERs) as well as  a reasonable approximation for very  low incident power at the ER, the linear EH model is inaccurate in many  practical scenarios due to the \textit{nonlinearity} of the ER  rectenna circuit that consists of two main components, namely, an antenna for intercepting the RF waveforms from the air, and a diode rectifier for
converting the collected RF waveform to DC voltage.
Specifically, $\eta$ and hence the end-to-end energy transmission efficiency  typically increase with the  peak-to-average power ratio (PAPR) of the incident signal waveform at the ER as well as  its average RF  power, but with diminishing returns and eventually may get saturated or even degraded with the increase of signal peak amplitude/power  due to the diode reverse breakdown effect  \cite{Boshkovska:2015,7054689,5972612,9153166}. Therefore, power waveform design under the practical rectenna nonlinearity is essential to maximize  the efficiency  of RF-based  WPT systems.

By taking the rectenna nonlinearity into account,  \cite{Clerckx:2016b,7932461,8008785,9204454,8470248,moghadam2017waveform,jsac_waveform,9096618} have studied the power waveform design problem for far-field WPT  by adopting the multisine waveform structure for the transmit signal. In \cite{Clerckx:2016b,7932461,8008785,9204454,8470248},  the rectenna nonlinearity is approximately characterized  via the second and higher order terms in the truncated Taylor expansion of the diode output current in terms of the input voltage. Based on this model, \cite{Clerckx:2016b} proposes a successive convex programming (SCP) algorithm for a single-ER WPT system to approximately design the amplitudes of different frequency tones of the multisine power waveform in an iterative manner, where at  each iteration a geometric programming (GP) problem needs to be solved. In \cite{7932461}, a low-complexity heuristic algorithm is proposed for a single-ER WPT system to obtain a closed-form solution for the power waveform  without the need of  solving any optimization problem. In \cite{8008785}, a large-scale multi-ER WPT system (where the number of ET antennas and number of frequency tones are large) is considered. Specifically, \cite{8008785} considers two waveform design problems to maximize the weighted-sum of the output DC powers at all ERs  and to maximize the minimum output DC power among all ERs, respectively. In \cite{9204454}, power waveforms for  a single-ER   WPT system are designed based on two  combining strategies at the ER, namely DC and RF combining, respectively. In \cite{8470248}, a new form of power waveform for a single-ER WPT system is developed based on phase sweeping transmit diversity. 

Different from  \cite{Clerckx:2016b,7932461,8008785,9204454,8470248},  \cite{moghadam2017waveform,jsac_waveform,9096618} consider the nonlinear EH model based on a circuit analysis that  captures the rectenna nonlinearity without relying on Taylor approximation.  In particular, \cite{moghadam2017waveform} studies a power waveform optimization problem to maximize the output  DC power of a single-input single-output (SISO) WPT system, and  proposes two low-complexity solutions based on frequency-domain maximal ratio transmission (MRT) and SCP, respectively. In \cite{jsac_waveform,9096618}, the output DC power of an ER is derived in a closed form by applying the Lambert $W$ function to the EH model proposed in \cite{moghadam2017waveform}, and then,  power waveform design problems are studied to  facilitate the energy transfer in a multi-ER WPT system.

In this paper, we  study  the power waveform design problem for a general WPT system  with a multi-antenna ET and multiple single-antenna ERs and fully exploit the nonlinear EH model in the multisine waveform design to maximize the end-to-end efficiency. Different from  prior works  \cite{Clerckx:2016b,7932461,8008785,9204454,8470248,moghadam2017waveform,jsac_waveform,9096618}, we first develop a refined  rectifier  model based on  circuit analysis that more accurately captures the characteristics of a practical rectenna as compared to those in \cite{moghadam2017waveform,jsac_waveform,9096618}. With  this new rectenna  model, we then formulate the power  waveform   optimization problems for the single-ER and  multi-ER cases, respectively,  to maximize the harvested power by the ER(s). The main contributions of this paper are summarized as follows.

\begin{itemize} 
	\item First, we propose a refined  nonlinear current-voltage model of the diode in the ER rectifier, and accordingly  derive a new closed-form expression for the output  DC voltage  in terms of the incident RF signal. Different from the existing EH models, the  new rectenna model not only captures the nonlinearity of the energy harvesting power with respect to the RF signal  waveform, but also takes into account the intrinsic forward and reverse breakdown current-voltage (I-V) characteristics of the rectifier diode. Furthermore, by assuming that the capacitance of the low pass filter (LPF) of the rectenna is sufficiently large in practice, our new model shows that maximizing the DC output power is equivalent to maximizing the time average of an exponential function in terms of the received  signal waveform. 
	\item Next, based on this new rectenna  model and by adopting a multisine waveform structure for the transmit signal, we formulate a problem to  optimize the subcarrier selection (i.e., selecting $N$ subcarriers for WPT from a larger number of equally-spaced $U$  subcarriers over the total available bandwidth) and the power  waveform with  selected subcarriers  for maximizing the harvested  DC power by a single ER. A two-step algorithm  is  proposed for finding an approximate solution to the formulated problem, which is non-convex  and thus difficult to be solved  optimally. First, we show that the optimal subcarrier selection for the power waveform can be obtained in  closed-form. Second, with fixed subcarrier selection, we propose two low-complexity algorithms for designing the  amplitude or power allocation over the selected subcarriers. The first algorithm is based on a closed-form solution, which corresponds to the well-known maximal ratio transmission (MRT) over the frequency subcarriers. While in the second algorithm, we employ SCP to iteratively allocate power over the selected subcarriers, where at each iteration the problem is approximated by a convex quadratically constrained linear programming (QCLP), for which the optimal solution is derived in closed form and thus can be efficiently computed.
	\item Furthermore,  we extend the single-ER power waveform optimization problem to the  general multi-ER case,  to maximize the weighted sum of the harvested  DC  power by all  ERs. To solve this more challenging problem than  that in the single-ER case, we propose a difference-of-convex-functions programming (DCP)-based iterative algorithm, where at each iteration the problem is approximated by a convex quadratically constrained quadratic  programming (QCQP) problem that can be efficiently solved. 
	\item Finally, simulation  results unveil substantial performance gains achieved by the proposed power waveform optimization solutions based 	on the refined  rectenna  model as compared to  existing schemes based on the conventional model. In particular, the newly proposed  subcarrier selection optimization leads to significantly better performance in both single-ER and multi-ER WPT cases, as compared  to the conventional designs without subcarrier selection (i.e., by setting $U=N$ directly and optimizing the  power waveform with $N$ equally-spaced subcarriers). Besides, it is also shown that the consideration of rectenna power saturation (or diode reverse breakdown) in our refined model yields  significantly improved  performance in a multi-ER WPT system as it helps to more flexibly balance the harvested DC  power at different ERs, especially when their  channels are in heterogeneous conditions.   
\end{itemize}

It is worth noting that this paper is an extension of our prior work  \cite{moghadam2017waveform}. In  \cite{moghadam2017waveform}, frequency-domain MRT and SCP methods are proposed to design the power waveform with equally-spaced  subcarriers for a point-to-point single-input-single-output (SISO)  WPT system (i.e., both the ET and ERs are equipped with a single antenna); while in this paper, we consider the more general multi-antenna ET and multiple ERs. Moreover, \cite{moghadam2017waveform} does not consider  the subcarrier selection and the  rectenna power saturation (or diode reverse breakdown) in the rectifier model, which are newly introduced in this paper.
 
The rest of this paper is organized as follows. Section II introduces the WPT system model.  Section III presents  the refined  rectenna model based on  circuit analysis. Section IV  formulates the new  multisine power waveform optimization problem for the single-ER case based on the refined rectenna  model  and presents the two-step approach for finding an approximate solution. Section V extends the single-ER power waveform optimization problem to the  general multi-ER case and presents the DCP-based iterative algorithm to solve it. Section VI presents simulation results to evaluate the effectiveness of the proposed rectenna  model and power waveform design  algorithms. Finally, we conclude the paper in Section VII.

\textit{Notations:} In this paper, scalars are denoted by italic letters, vectors and matrices are denoted by bold-face lower-case and upper-case letters, respectively. For a complex-valued vector $ \mathbf x $, $ \| \mathbf x \|_p $ and $ \mathbf x^H $ denote its $\ell_p$-norm and conjugate transpose, respectively. Scalar $x_i$ denotes the $i$-th element of vector $\mathbf x$.  $ \mathbb{C}^{x \times y} $  denotes the space of $ x \times y $ complex-valued matrices.  $\lceil a \rceil$ denotes the smallest integer greater than or equal to $a$. $ j $ denotes the imaginary unit, i.e., $ j^2 = -1 $. $ \mathrm{Re}\{\cdot\} $ represents the real part of a complex number.

\section{System Model}\label{Section:system_model}
As shown in Fig. \ref{fig:System_Model},  we consider a multi-user WPT system where an ET is employed to  deliver energy wirelessly to $K\ge 1$  ERs, denoted by the set ${\cal K}=\{1,\ldots, K\}$. We assume that the ET has $M\ge 1$ antennas, denoted  by $m \in {\cal M}= \{1,\ldots,M\}$, and each ER has  a single antenna. Furthermore, we assume that the available frequency band for WPT is  given by  $[f_{\min},~ f_{\max}]$ in Hertz (Hz) with $B=f_{\max}-f_{\min}$  denoting the bandwidth. Let $ {\cal U}= \{1,\ldots,U\}$ denote the set of equally spaced orthogonal subcarriers that span over the  frequency band. Accordingly, the $u$-th subcarrier frequency is given by $f_u=f_{0}+(u-1)\Delta_u$, where $f_0\ge f_{\min}$ and $\Delta_u>0$ are designed such that $f_0/ \Delta_u$ is an integer and $f_0+(U-1)\Delta_u \le f_{\max}$. In particular, we  set $\Delta_u= B/U$ and $f_0=\lceil f_{\min}/\Delta_u\rceil \Delta_u$.  

\begin{figure}[t!] \centering
	\centering \vspace{0mm}
	\includegraphics[trim = 0mm 0mm 0mm 0mm, clip,width=0.7\linewidth]{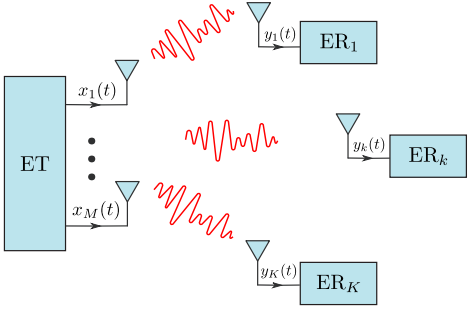}
	\caption{A multi-user multi-antenna wireless power transfer system.} \vspace{-0.5cm}
	\label{fig:System_Model} 
\end{figure}

We adopt the practical  multisine signal  structure for the power waveform at the ET   \cite{Clerckx:2016b}. Let  $\mathbf{ \tilde s}_m = [\tilde s_{m,1},\dots,\tilde s_{m,U}]^T $ be a vector constituting   sinewave amplitudes and phases over $ U $ subcarriers on antenna $m$, where $\tilde{s}_{m,u}=s_{m,u}\exp(j\phi_{m,u})$,  $s_{m,u}\ge 0$ and $ \phi_{m,u}\in [0,~2\pi)$ denote the amplitude and phase of the $u$-th sinewave  at frequency $f_u$, respectively. Hence,  the transmit signal at time $t$ on antenna $m$ is expressed as  
\begin{equation}
\label{x_m}
x_m(t)=\operatorname{Re}\Bigg\{\sum_{u \in  \cal U} \sqrt{2} \tilde{s}_{m,u} \exp(jw_ut) \Bigg\},
\end{equation}
where $w_u=2\pi f_u, \forall u \in \mathcal U $. Note that the number of sinewaves used for WPT in practice is generally  limited to $ N \leq U$ due to the practical PAPR constraint at each antenna of the ET \cite{Clerckx:2016b}. We thus have the following subcarrier selection constraint at each transmitting antenna: 
\begin{align} \label{l_0}
\|\mathbf{ \tilde s}_m\|_0 \leq N,  m =1,\dots, M.
\end{align}
Moreover, suppose that the ET has a total transmit power budget $P_T$, i.e.,
\begin{align} \label{eq:Sum_Power_constraint}
\dfrac{1}{T}\int_{T} \sum_{m \in \mathcal M} x^2_m(t) dt=  \sum_{m \in \mathcal M} \|\mathbf{ \tilde s}_m\|^2_2 \leq P_T,
\end{align}
where $T=1/\Delta_u$ denotes the period of $x_m(t)$.

The transmitted multisine waveform  propagates through the wireless channel with generally  multiple delayed  paths to each ER.  Let ${\cal L}_{k}$ denote the set consisting of all the paths between  the ET and  ER $k$.  Furthermore, we denote the delay, amplitude, and phase at the $m$-th antenna for the $l$-th path, $l \in {\cal L}_{k}$, by  $\tau_{l,m,k} >0$, $\alpha_{l,m,k} >0$, and $\xi_{l,m,k} \in [0, ~2\pi)$, respectively. 
The incident RF signal at ER $k$ is thus given by 
\begin{align} \label{eq:y(t)}
y_k(t)&= \operatorname{Re} \left\{\sum_{m \in \cal M} \sum_{u \in \cal U} \sum_{l \in {\cal L}_{k}}\sqrt{2} s_{m,u}\alpha_{l,m,k}  \exp(j \Delta_{l,m,k}(t))\right\}  \nonumber \\ &=\operatorname{Re} \left\{\sum_{m\in \cal M} \sum_{u\in \cal U}\sqrt{2}  \tilde{h}_{k,m,u} \tilde{s}_{m,u} 
\exp(j w_u t)\right\}, 
\end{align}
where $  \Delta_{l,m,k}(t)=w_u(t-\tau_{l,m,k})+\xi_{l,m,k}+\phi_{m,u}$ and  $ \tilde{h}_{k,m,u} = h_{k,m,u} \exp(j \psi_{k,m,u})= \sum_{l\in{\cal L}_{k}} \alpha_{l,m,k}  \exp(j ( -w_u  \tau_{l,m,k} + \xi_{l,m,k})) $, which denotes  the channel frequency response between antenna $m$ and ER  $k$ at frequency $f_u$, with $h_{k,m,u} \geq 0$ and $ \psi_{k,m,u} \in [0,~2\pi)$  respectively representing the amplitude and  phase. Let $ \mathbf{ \tilde h}_{k,m} = [\tilde{h}_{k,m,1}, \dots, \tilde{h}_{k,m,U}]^T  $. In this paper, we assume that the channel state information (CSI), i.e., $\{ \mathbf{ \tilde h}_{k,m} \}_{m \in \mathcal M}, $  are known at the ET (e.g., via the reverse-link channel training and estimation \cite{7089273})  and remain constant for a given period of time for WPT.  Based on such CSI, we  design  $\{ \mathbf{ \tilde s}_m \}_{m \in \mathcal M}$ in Sections \ref{Sec:single-RX WFormOPT} and \ref{Sec:multi-RX WFormOPT} for the single-ER and multi-ER cases, respectively.

\section{Refined Rectenna Model}\label{Section:Analytical_Model_Rectenna}
In this section, we present a refined   rectenna circuit model  to capture its intrinsic nonlinearity accurately and derive its output DC voltage as a function of the incident RF signal. We assume that the electric circuits of all ERs are identical,  and thus the ER  index $k$ is omitted in this section. 

\begin{figure}[t!] 
	\centering \vspace{0mm}
	\includegraphics[trim = 0mm 0mm 0mm 0mm, clip,width=0.7\linewidth]{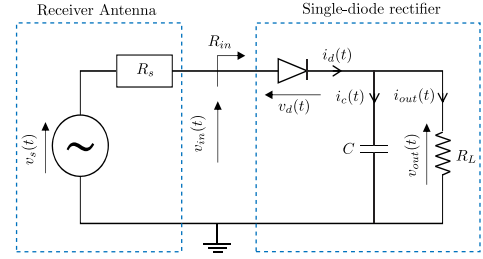}
	\caption{Rectenna circuit model.} \vspace{-0.5cm}
	\label{fig:first} 
\end{figure}

\subsection{Rectenna Circuit Analysis}\label{Sect:Rectenna_eq_circuit}
As shown in Fig. \ref{fig:first}, a typical rectenna consists of two main components, namely, an antenna for intercepting the RF waveforms from the air,  and a single-diode  rectifier for converting the collected RF waveform to DC for direct use or battery charging. 
The antenna is  commonly modelled as a voltage source $v_{s}(t)$ in series with a  resistance $R_{s}>0$, where $v_{s}(t)=2  \sqrt{R_{s}} y(t)$ \cite{Clerckx:2016b}. On the other hand, the rectifier consists of a single diode, which is   nonlinear,  followed by a low pass filter (LPF) connected to an electric load, with the load resistance  denoted by  $R_{L}>0$. As shown in Fig. \ref{fig:first}, we denote $R_{in}>0$ as the equivalent input resistance of the rectifier.  With perfect impedance matching, i.e., $R_{s}=R_{in}$, the input voltage of the rectenna, denoted by $v_{in} (t)$, is obtained as $v_{in}(t)=v_{s}(t)/ 2= \sqrt{R_{s}} y(t)$. Since $y( t)$ is periodic, it follows that  $v_{in} (t)$ is also periodic with the same period $T$.

\begin{figure}[t!]
	\centering
	\includegraphics[width=0.75\linewidth]{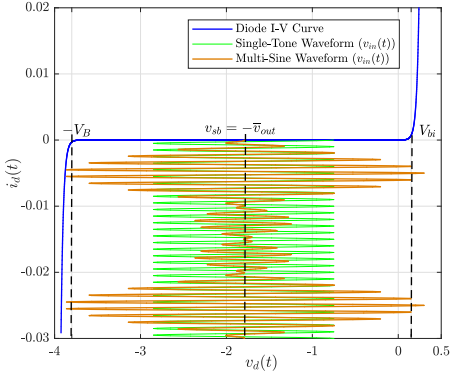}
	\caption{Self-bias behaviour and  reverse breakdown of the diode.}  
	\label{fig3:first} 
\end{figure}

In this paper, as compared to  prior works  \cite{Clerckx:2016b,7932461,8008785,9204454,8470248,moghadam2017waveform,jsac_waveform,9096618}, we consider a more accurate nonlinear diode model that captures both the forward and  reverse breakdown I-V characteristics of the rectifier diode. Let $i_d(t)$ and $v_d(t)$ denote the current passing through the diode and the voltage drop across it, respectively. We then have 
\begin{multline} \label{eq:Diode_iv}
i_{d}(t)=I_0\left(\exp\Big(\dfrac{v_d(t)}{\eta V_0}\Big)-1\right) \\ - I_{BV}\left(\exp \Big(\dfrac{-(v_d(t)+V_B)}{\eta V_0} \Big)\right),
\end{multline}
where $\eta>1$,  $V_0>0$, $I_0>0$, $ I_{BV} >0 $, and $ V_B > 0$ denote the diode ideality factor, thermal voltage, reverse bias saturation current, breakdown saturation current, and reverse breakdown voltage, respectively \cite{app_note_diode}. In particular, the first exponential term governs the diode behaviour when $ v_d(t) \geq 0 $, the second exponential term governs the diode behaviour when $ v_d(t) < -V_B  $, and both of the exponential terms  are negligible and $ i_{d}(t) \approx -I_0 $ when $ -V_B < v_d(t) < 0 $.

\subsection{Output  DC Power}\label{Section:DC_Voltage} 
By applying Kirchhoff's circuit laws to the electric circuit of the rectifier shown in Fig. \ref{fig:first}, we have the following set of equations:
\begin{align} 
&i_d(t)=i_c(t)+ i_{out}(t),\label{eq:Krichhof1}\\ 
&i_c(t)=C \dfrac{d v_{out}(t)}{dt},\label{eq:Krichhof2}\allowdisplaybreaks \\
&v_{d}(t)=v_{in}(t)-v_{out}(t),\label{eq:Krichhof3}\\
&v_{out}(t)= R_L i_{out}(t), \label{eq:Krichhof4}
\end{align}
where $v_{out}(t)$, $i_{out}(t)$, $C$, and $i_c(t)$ are the rectifier output voltage, load current, capacitance of the LPF, and the current passing through it, respectively. After some manipulations  based on  (\ref{eq:Diode_iv})--(\ref{eq:Krichhof4}), the relationship between the input and output voltages of the rectifier is obtained as
\begin{multline} \label{eq:vout(t)-vin(t)}
\hspace{-2mm} I_0\left(\exp\Big(\dfrac{v_{in}(t)-v_{out}(t)}{\eta V_0}\Big)-1\right) \\ - I_{BV}\left(\exp \bigg(\dfrac{-(v_{in}(t)-v_{out}(t)+V_B)}{\eta V_0} \bigg)\right) \\ =C\dfrac{d v_{out}(t)}{dt} + \dfrac{v_{out}(t)}{R_L}.\hspace{-2mm}
\end{multline}

Let  $v_{out}(t)=\overline{v}_{out}+ \tilde{v}_{out}(t)$, with  $\overline{v}_{out}$ and $\tilde{v}_{out}(t)$ denoting the DC and AC (alternating current)  components of the output voltage, respectively.  
Since the input voltage of the rectifier   $v_{in} (t)$ is periodic,  it follows from  (\ref{eq:vout(t)-vin(t)}) that $\tilde{v}_{out}(t)$ is also periodic with the same period $T$, i.e., $\tilde{v}_{out}(t)=\tilde{v}_{out}(t+kT)$ for any integer $k$. Moreover,  the average value of $\tilde{v}_{out}(t)$  is zero, i.e., $\tfrac{1}{T}  \int_{T} \tilde{v}_{out}(t) dt=0$. 
Next, by averaging both sides of (\ref{eq:vout(t)-vin(t)}) over the period $T$, we obtain 
\begin{eqnarray} \label{eq:vout(t)-vin(t)-int}
\lefteqn {I_0 \left(\dfrac{1}{T} \exp\Big(-\dfrac{\overline{v}_{out}}{\eta V_0}\Big) \int_{T}  \exp\Big(\dfrac{v_{in}(t)-\tilde{v}_{out}(t)}{\eta V_0}\Big)dt  -1\right) } \nonumber \\ &&  \hspace{-1cm}  { - I_{BV} \left(\dfrac{1}{T} \exp\Big(\dfrac{\overline{v}_{out} - V_{B} }{\eta V_0}\Big) \int_{T}  \exp\Big(\dfrac{-v_{in}(t)+\tilde{v}_{out}(t) }{\eta V_0}\Big)dt \right)} \nonumber \\  & =&   \dfrac{C}{T} \int_{T} \dfrac{d \tilde{v}_{out}(t)}{dt} dt + \dfrac{1}{T R_L} \int_{T}  \big(\overline{v}_{out}+\tilde{v}_{out}(t)\big) dt \nonumber \\
& = &  C\big(\tilde{v}_{out}(T) - \tilde{v}_{out}(0)\big)+ \dfrac{1}{T R_L} \big(\ T \overline{v}_{out} +0\big) \nonumber \\  & = &  \dfrac{\overline{v}_{out}}{R_L}.
\end{eqnarray}
By assuming that the capacitance $C$ of the LPF is sufficiently large\footnote{It should be noted that to keep  the output DC voltage $\overline{v}_{out}$ constant over time, the capacitance of the LPF should be sufficiently large such that $C R_L \gg T$, which can be explained as follows. 	When the diode is reversely biased,  the output voltage is governed by the discharging law of the capacitor of the LPF, and is proportional to  $\exp(-\tfrac{t}{C R_L})$. 	In this case, by setting e.g.  $C=50 T/R_L$,  the normalized output voltage fluctuation (i.e., divided by the peak of the output voltage) is obtained as $1-\exp(-0.02)=1.98 \%$, which is reasonably small as practically desired. }, the output AC voltage of the rectifier is small, i.e.,   $\tilde{v}_{out}(t)\approx 0$ \cite{Clerckx:2016b}.   
Hence,  (\ref{eq:vout(t)-vin(t)-int}) can be simplified as   
\begin{multline} \label{eqxx}  
\exp\Big(\dfrac{\overline{v}_{out}}{\eta V_0}\Big)\left(1+\dfrac{\overline{v}_{out}}{R_L I_0}\right) =\dfrac{1}{T} \int_{T}  \exp\Big(\dfrac{v_{in}(t)}{\eta V_0}\Big) dt \\- \dfrac{1}{T}  \left(\frac{I_{BV}}{I_0} \right) \exp\Big(\dfrac{2\overline{v}_{out} - V_{B} }{\eta V_0}\Big) \int_{T}  \exp\Big(\dfrac{-v_{in}(t)}{\eta V_0}\Big) dt. 
\end{multline}
Since $ v_{in}(t) = \sqrt{R_s}y(t)$ and $ \dfrac{1}{T} \int_{T} \exp\Big(\dfrac{\sqrt{R_s}y(t)}{\eta V_0}\Big) dt = \dfrac{1}{T} \int_{T} \exp\Big(\dfrac{-\sqrt{R_s}y(t)}{\eta V_0}\Big) dt $, \eqref{eqxx} can be further simplified as
\begin{align} \label{eqyy}  
\frac{\exp\Big(\dfrac{\overline{v}_{out}}{\eta V_0}\Big)\left(1+\dfrac{\overline{v}_{out}}{R_L I_0}\right)}{1- \left(\frac{I_{BV}}{I_0} \right)\exp\Big(\dfrac{2\overline{v}_{out} - V_{B} }{\eta V_0}\Big)}
&=\dfrac{1}{T} \int_{T} \exp\Big(\dfrac{\sqrt{R_s}y(t)}{\eta V_0}\Big) dt. 
\end{align}
In the sequel of this paper, we aim to design the multisine power waveform based on the new equation  derived in \eqref{eqyy}, which essentially captures the relationship between the rectenna output DC voltage and the incident RF signal.

Since the right hand side (RHS) of \eqref{eqyy} is non-negative, we should  have ${\overline{v}_{{out}}} <  {\frac{1}{2} \eta V_0 \ln\Big(\frac{I_0}{I_{BV}}\Big) + \frac{V_{B}}{2}}$. Accordingly, the maximum attainable DC voltage, denoted by $\overline{v}_{{out}}^\ast$, is given by
\begin{align} \label{denominator}
\overline{v}_{{out}}^\ast = {\frac{1}{2} \eta V_0 \ln\Big(\frac{I_0}{I_{BV}}\Big) + \frac{V_{B}}{2}}. 
\end{align}
As shown in Fig. \ref{fig3:first}, the total voltage applied to the diode is equal to $  v_{in}(t) + v_{sb} $, where $ v_{sb} = \frac{1}{T} \int_{T} v_{d}(t)   dt = \frac{1}{T} \int_{T} (v_{in}(t) - v_{out}(t))  dt \approx -\overline{v}_{out}$, is the mean value of the signal  (or the self-bias voltage) across the diode. Thus, when $ \overline{v}_{out} $ increases, $ v_{sb} $ will shift towards $-V_B$. Once the amplitude of the signal across the diode reaches the diode breakdown voltage $ V_B $, a significant amount  of the reverse current will start to flow through the diode in the negative cycles of the input signal, and the reverse breakdown occurs. In this case, regardless of any increase in the input signal, $ \overline{v}_{out} $ will be fixed to its maximum value  $\overline{v}_{{out}}^\ast$. It should be noted that $ \frac{1}{2} \eta V_0 \ln\Big(\frac{I_0}{I_{BV}}\Big)  $ in \eqref{denominator} typically takes a small negative value as  $ I_0  \ll I_{BV} $ and $ V_0 \approx 26$ mV at the room temperature. Thus, we have 
\begin{align}
\overline{v}_{{out}}^\ast \approx \frac{V_B}{2}.
\end{align} 
Accordingly, it is not difficult to show that any power waveform  that satisfies
\begin{align}
v_{in}(t)=\sqrt{R_s}y(t) &\leq \frac{V_B}{2}, \label{saturation}
\end{align}
for $0\geq t \geq T$, guarantees that it does not exceed the diode reverse breakdown voltage. Moreover, the diodes used in rectenna circuits usually feature very low barrier voltage  ($ V_{bi} $ in Fig. \ref{fig3:first}), which leads  to a desired  biasing behaviour at very small input power levels. However, such devices have very low values for $ V_B $, which  limits the maximum output power  \cite{diode_3}.

It can be seen from \eqref{eqyy} that its RHS, denoted by \vspace{-1mm}
\begin{equation} \label{rhs}
\Psi_{\mathrm{RHS}}(\{\mathbf{ \tilde s}_m\}) = {  \dfrac{1}{T} \int_{T} \exp\left(\dfrac{\sqrt{R_s} y(t) }{\eta V_0}\right) dt },
\end{equation}
is a function of $\{\mathbf{ \tilde s}_m\}_{m =1}^M$, while its left hand side (LHS) is denoted by \vspace{-1mm}
\begin{equation} \label{lhs}
\Psi_{\mathrm{LHS}}(\overline{v}_{out}) = \frac{\exp\left(\dfrac{\overline{v}_{out}}{\eta V_0}\right)\left(1+\dfrac{\overline{v}_{out}}{R_L I_0}\right)}{1- \left(\frac{I_{BV}}{I_0} \right)\exp\left(\dfrac{2\overline{v}_{out} - V_{B} }{\eta V_0}\right)},
\end{equation}
as a function of $ \overline{v}_{out} $. 
Moreover, $\Psi_{\mathrm{LHS}}(\overline{v}_{out}) $ is monotonically non-decreasing for $\overline{v}_{out} \in [0,~\overline{v}_{out}^\ast]$. Thus,  maximizing $ \overline{v}_{out} $ is equivalent to maximizing $\Psi_{\mathrm{RHS}}(\{\mathbf{ \tilde s}_m\}) $ by optimizing $\{\mathbf{ \tilde s}_m\}_{m =1}^M$.  Then, with the optimized $\{\mathbf{ \tilde s}_m\}_{m =1}^M$, we can obtain the corresponding $\overline{v}_{out}$ by using a simple  bisection search, which is omitted for brevity.

\subsection{Performance Validation }\label{Section:Performance characterization}

To validate and  draw  insights into  the nonlinear rectenna model derived  in \eqref{eqyy}, a circuit simulation was carried out, and the results are illustrated in Fig. \ref{fig:sim}. By varying the RF  power of the received signal $y(t)$  and the  number of frequency tones $ N $,  output DC power was calculated with $N=U$ (equal power allocation over all frequency tones is employed under real valued channel parameters with unit magnitudes).  The diode parameters follow the values in HSMS-285x datasheet \cite{datasheet}. Some important observations can be made from Fig. \ref{fig:sim}:
\begin{enumerate}  
	\item The output power is  monotonically non-decreasing with the input  power, as expected. 
	\item The nonlinear rectifier gain generally increases with  $ N $ for a fixed input power.
	\item As $ N $ increases, the output  power increases at low input power values. This is due to the greater ability of  multi-sine waveforms to overcome the built-in voltage  $V_{bi}$ of the diode as compared to the single-tone waveform (see Fig. \ref{fig3:first}). 
	\item For larger $ N $, the output power saturation is reached with smaller input RF power. 
	\item The nonlinear models given in \cite{Clerckx:2016b,7932461,8008785,9204454,8470248,moghadam2017waveform,jsac_waveform,9096618} follow a similar curve. However, they fail to capture the power saturation effect. 
\end{enumerate}
It is worth noting that in  \cite{Boshkovska:2015}, the authors have shown  a nonlinear model for wireless energy harvesting with similar  monotonically non-decreasing  and  power saturation effects as those shown in Fig. \ref{fig:sim}, but without the analytical results given in \eqref{eqyy}.  

\begin{figure}[t!] 
	\centering
	\includegraphics[trim = 0mm 0mm 0mm 0mm, clip,width=0.75\linewidth]{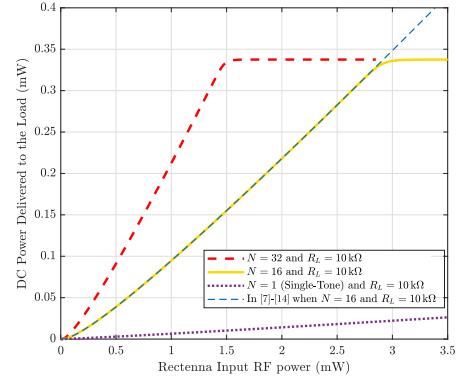}
	\caption{The behaviour of the  nonlinear rectenna model given in \eqref{eqyy}.}  
	\label{fig:sim} 
\end{figure}

\section{Single-ER Waveform Optimization} \label{Sec:single-RX WFormOPT}

In this section, we consider the special case  with a single ER and study the power  waveform optimization problem  in order to draw essential insights into the optimal waveform design.  For brevity, the ER index $k$ is omitted in this section.

\subsection{Problem Formulation} \label{sec:Single-ER_PF}

Since the ER power saturation does not adversely effect for the output DC power of a single-ER system, we do not consider the constraint in \eqref{saturation}. Accordingly, we optimize  $ \{\mathbf{ \tilde s}_m\}_{m =1}^M$ to maximize the output  DC voltage at the single ER, subject to the  maximum transmit sum-power and the  subcarrier selection constraints. The problem is formulated as 
\begin{align} 
\mathrm{(P0)}:  
\mathop{\mathtt{max}}_{\{ \mathbf{ \tilde s}_m \}}~& \Psi_{\mathrm{RHS}}(\{\mathbf{ \tilde s}_m\})  \label{eq:P0_Obj} \\
\mathtt{s.t.}~& \sum_{m =1}^M \|\mathbf{ \tilde s}_m\|^2_2 \leq P_T, \label{eq:P0_C1} \\
~& \|\mathbf{ \tilde s}_m\|_0 \leq N,  m =1,\dots, M. \label{eq:P0_C2}
\end{align}
In (P0),  the objective function and the constraint in (\ref{eq:P0_C1}) are both convex functions  over $\mathbf{ \tilde s}_m$'s. However, maximizing a convex function over a convex set is non-convex in general. Moreover, the cardinality constraints in \eqref{eq:P0_C2} that select  $N$ out of $U$ subcarriers at each transmit antenna $m$, is non-convex and discontinuous, which makes (P0) more difficult to solve. Nevertheless, it is easy to show  that at the optimal solution to (P0), all $\cos(\cdot)$ terms in \eqref{eq:y(t)} should  be added constructively. As such, we need to set $\phi_{m,u}=-\psi_{m,u}$,  $\forall m,u$. With the above optimal phase design, we can focus on optimizing the sinewave amplitudes and subcarrier selection   by considering  the following equivalent  problem of (P0).
\begin{align} 
\mathrm{(P1)}:  
\mathop{\mathtt{max}}_{\{ \mathbf{  s}_m \}}~& \Psi_{\mathrm{RHS}}(\{\mathbf{  s}_m\})  \label{eq:P1_Obj} \\
\mathtt{s.t.}~& \sum_{m =1}^M \|\mathbf{  s}_m\|^2_2 \leq P_T, \label{eq:P1_C1} \\
~& \|\mathbf{  s}_m\|_0 \leq N,  m =1,\dots,M, \label{eq:P1_C2}
\end{align}
where $ \mathbf{  s}_m = [ s_{m,1},\dots, s_{m,U}]^T  $ denotes the sinewave amplitudes over $ U $ subcarriers at antenna $m$, with  $s_{m,u}\ge 0,\forall u \in \{1,\dots,U\}$. In the following, we propose a two-step algorithm  for finding a suboptimal  solution to problem (P1). Specifically, we first consider the subcarrier selection problem, and then investigate the  sinewave amplitude optimization with fixed subcarrier selection.

\subsection{Subcarrier Selection  Optimization} \label{sub_allocation}

Let $p_u\stackrel{\vartriangle}{=}\sum_{m = 1}^M s_{m,u}^2$,  with $\sum_{u = 1}^U p_u \le P_T$, denote the total power allocated over all transmit antennas at subcarrier $u$ subject to the subcarrier selection constraints in  \eqref{eq:P1_C2}. 
\begin{lemma} \label{Lemma:Spt-MRT}
	For any given  $p_u$'s at different subcarriers,  the optimal sinewave amplitude design over transmit antennas is given by $s_{m,u}= \frac{\sqrt{p_u}  h_{m,u}}{\sqrt{\sum_{m =1 }^M h_{m,u}^2}}$, $\forall m$. 
\end{lemma}
\begin{proof}
Let $\xi \ge 0$ denote the Lagrange dual variable corresponding to (\ref{eq:P1_C1}). Since the given power allocations $\{p_u\}_{u=1}^U$ are subject to the subcarrier selection, we omit the constraints in  \eqref{eq:P1_C2}. The Lagrangian of (P1) is thus given by
\begin{multline} 
\hspace{-2mm}\mathcal{L}= \dfrac{1}{T} \int_{T} \exp\Big(\dfrac{\sqrt{2R_{s}}  \sum_{m=1}^{M} \sum_{u=1}^{U}  h_{m,u} s_{m,u} \cos (w_u t)}{\eta V_0}\Big) dt \\ - \xi \left( \sum_{m=1}^{M} \sum_{u=1}^{U} s_{m,u}^2- P_T \right).
\end{multline}
One of the Karush–Kuhn–Tucker (KKT) conditions to (P1) is that the derivative of the  Lagrangian with respect to the primal variable $ s_{m,u} $ must be zero at the optimal point, i.e.,
\begin{multline}
\dfrac{\sqrt{2R_s}}{T\eta V_0} h_{m,u}  \int_{T} \cos (w_u t) \exp\Big(\dfrac{ \splitfrac{\sqrt{2R_{s}}  \sum_{m=1}^{M} \sum_{u=1}^{U} }{  h_{m,u}  s_{m,u} \cos (w_u t)}} {\eta V_0}\Big) dt \\ - 2 \xi  s_{m,u}=0,~\forall m,u. \label{eq:KKT-S1-1}
\end{multline}
For any given $u$, as the integral term is common for different antenna $m$, it follows that   $s_{m,u}/s_{{m'},u}=h_{m,u}/h_{{m'},u}$ at $f_u$. 
Let $p_u \ge 0$ denote the total power allocated over all transmit antennas at  $f_n$, i.e., $p_u=\sum_{m=1}^M s_{m,u}^2$,  which is subject to  $\sum_{u =1}^U p_u \le P_T$. Accordingly, for each frequency tone $u$, it follows that $s_{m,u}= \frac{\sqrt{p_u}  h_{m,u}}{\sqrt{\sum_{m =1}^M h_{m,u}^2}} $, $\forall m$, is the \textit{unique} solution to (\ref{eq:KKT-S1-1}). The proof  is thus completed.	
\end{proof}

It is also inferred from the above proof  that at each subcarrier, the optimal power allocation over  transmit  antennas follows the well-known   maximal ratio transmission (MRT) principle \cite{viswanathtse}. We thus term the solution given in Lemma \ref{Lemma:Spt-MRT} as \textit{spatial-MRT}. 
For  convenience, we introduce a change of variable as
$\mathbf x = [x_1,\dots,x_U]^T $, with $x_u=\sqrt{p_u}$,  $\forall u \in \{1,\dots,U\}$. With the spatial-MRT solution, \eqref{eq:P1_Obj} can be  equivalently rewritten  as
\begin{equation} \label{mrt_obj}
 \Psi_{\mathrm{RHS}}(\mathbf{x}) = \dfrac{1}{T} \int_{T} \exp\Bigg(\dfrac{ {\sqrt{2R_{s}}  \sum_{u =1}^U  b_{u} x_u }{ \cos (w_u t)}}{\eta V_0}\Bigg) dt,
\end{equation}
where  $b_u = {\sqrt{\sum_{m =1}^M h_{m,u}^2}}, \forall u$. From \eqref{mrt_obj}, it is observed that $b_u$ is the effective channel gain for the $u$-th subcarrier  over all transmit  antennas. The optimal subcarrier selection  is thus given in the following proposition (for which the proof is straightforward and thus omitted for brevity).
\begin{proposition} \label{sub_allocation_pro}
Let $\mathbf b = [b_1,\dots,b_U]^T$ with $b_u = {\sqrt{\sum_{m =1}^M h_{m,u}^2}}, \forall u$. Then, the index set of the optimal subcarrier selection solution to (P1), denoted by $\mathcal {U_N} $, is given by the indices of the $N$  largest components of $\mathbf b$.
\end{proposition}

\subsection{Sinewave Amplitude Optimization} \label{amp_opt}

We now consider the sinewave amplitude optimization with selected    subcarriers. Let $ \underline{\mathbf b}=[\underline b_1,\dots,\underline b_U]^T $, with 
\begin{align} 
\underline b_u &= 
\begin{cases} 
{\sqrt{\sum_{m =1}^M h_{m,u}^2}} & \text{if  } ~~u\in \mathcal {U_N} \\
0      & \text{otherwise  }
\end{cases},
\end{align}
and 
\begin{equation} \label{mrt_obj22}
\underline{\Psi}_{\mathrm{RHS}}(\mathbf{x}) = \dfrac{1}{T} \int_{T} \exp\Bigg(\dfrac{ {\sqrt{2R_{s}}  \sum_{u =1}^U  \underline b_{u} x_u }{ \cos (w_u t)}}{\eta V_0}\Bigg) dt. \nonumber
\end{equation}
We  equivalently rewrite (P1) in terms of $\mathbf x$ as  follows. 
\begin{align} 
\mathrm{(P1-EQ)}:  
\mathop{\mathtt{max}}_{ \mathbf{x} }~& \underline \Psi_{\mathrm{RHS}}(\mathbf{x})  \label{eq:P1_Simpl_Obj} \\
\mathtt{s.t.}~& \|\mathbf{x}\|^2_2 \leq P_T. \label{eq:P1_Simpl_C1}
\end{align}
The problem (P1--EQ) is non-convex and difficult to be optimally solved in general, for which we propose two efficient  suboptimal  solutions.

\subsubsection{SCP-based Solution} \label{sec:SCP_Sol} 
SCP is an iterative method to solve the problem of maximizing a convex  function (thus non-convex optimization problem) suboptimally. Specifically, at each iteration  $l=1,2,\ldots$, we approximate  the objective function in (P1--EQ) by a linear function using  its first-order Taylor series to form a convex optimization problem. Next, we set the values of decision variables $x_{u}$'s  for iteration $l+1$ as the optimal solution to the approximate problem at iteration $l$. The algorithm continues until a given stopping criterion is satisfied. In the following, we provide details of our SCP based algorithm for (P1--EQ).

Let $x_{u}^{(l)}$, $\forall u $, denote the  variables of (P1--EQ) at the beginning of iteration $l$. We approximate the objective function in (\ref{eq:P1_Simpl_Obj}) via its first-order Taylor series as  
\begin{align} \label{eq:Approx_Obj_P1}
 \underline\Psi_{\mathrm{RHS}}^{(l)}(\mathbf{x}) = \beta_{0}^{(l)} +\sum_{u=1}^U \beta_{u}^{(l)} \left(x_{u}-x_{u}^{(l)} \right),
\end{align}
with the coefficients given by 
\begin{align}
\beta_{0}^{(l)}&=\dfrac{1}{T} \int_{T} z^{(l)}(t)  dt, \label{eq:beta0}\\ 
\beta_{u}^{(l)}&=\dfrac{1}{T} \int_{T} \dfrac{\sqrt{2R_s}}{\eta V_0} \underline b_{u} \cos(w_u t) z^{(l)}(t) dt, \forall u, \label{eq:betan}
\end{align}
where  $z^{(l)}(t) \triangleq   \exp\big(\sqrt{2R_{s}}  \sum_{u=1}^U \underline b_{u} x_{u}^{(l)}  \cos (w_u t)/ (\eta V_0)\big)$. 
Since  the objective function in (\ref{eq:P1_Simpl_Obj}) is convex over $x_{u}$'s, the linear approximation given in (\ref{eq:Approx_Obj_P1}) is  its {global} under-estimator \cite{202}. The integrals in (\ref{eq:beta0}) and (\ref{eq:betan}) can be computed by using numerical integration techniques such as Newton-Cotes formula   \cite{Ube:1997}.
 
Next, we present the approximate problem of (P1--EQ) for each iteration $l$ as follows. 
\begin{align} 
\mathrm{(P1-EQ-\mathnormal{l})}:  
\mathop{\mathtt{max}}_{\mathbf x}~& \underline\Psi_{\mathrm{RHS}}^{(l)}(\mathbf{x}) \label{eq:P1-l_Obj} \\
\mathtt{s.t.}~& \|\mathbf x\|_2^2 \leq P_T. \label{eq:P1-l_C1} 
\end{align}
(P1--EQ$-\mathnormal{l}$) is a convex  quadratically constrained linear programming (QCLP), where the optimal solution is given by $x_{u}=  \sqrt{P_T}  {\beta}_{u}^{(l)}/{\beta}^{(l)}$, $\forall u$, with  ${\beta}^{(l)}=  \sqrt{\sum_{u=1}^U ({\beta}_{u}^{(l)})^2}$. Next, we set $x_{u}^{(l+1)}$'s using the solution of (P1--EQ$-\mathnormal{l}$).  
Accordingly, we compute ${\beta}_{0}^{(l+1)}$, and update  $\Delta_{{\beta}_{0}}=|{\beta}_{0}^{(l+1)}- {\beta}_{0}^{(l)}|/{\beta}_{0}^{(l)}$. 
Let $\epsilon>0$ denote a given stopping threshold. 
If $\Delta_{{\beta}_{0}} \le \epsilon$, then the algorithm terminates. Otherwise, if $\Delta_{{\beta}_{0}} >\epsilon$, then the algorithm will continue to the next iteration. The above iterative algorithm  is thus  named SCP-QCLP. Note that this algorithm cannot  guarantee to obtain the globally  optimal solution of  (P1--EQ) in general,  but can yield a point fulfilling its KKT conditions  \cite{202}.  Hence,  it will return (at least) a locally optimal solution to (P1--EQ). 

\subsubsection{Frequency-MRT} \label{sec:MRC_Sol}
Let $z(t)=\exp(\sqrt{2R_{s}}  \sum_{u=1}^U \underline b_{u} x_u \cos (w_u t)/(\eta V_0))$. One approximate way to solve (P1--EQ) is  by replacing the integral in (\ref{eq:P1_Simpl_Obj})  by the peak value of its integrand $z(t)$  over the period $T$, which is given by   $z(0)=\exp(\sqrt{2 R_{s}} \sum_{u=1}^U \underline b_{u} x_{u}   /(\eta V_0))$.  
Accordingly, \eqref{eq:P1_Simpl_Obj} is approximated as 
\begin{equation} \label{mrt_obj2}
\underline \Psi_{\mathrm{RHS}}(\mathbf{x}) \approx  \exp\Bigg(\dfrac{ {\sqrt{2R_{s}}  \sum_{u =1}^U \underline b_{u} x_u }{ }}{\eta V_0}\Bigg) 
\end{equation}
and the optimal  (frequency-MRT) solution to maximize the RHS of \eqref{mrt_obj2} is  given  in the following proposition. The proof is similar to that of  Proposition \ref{Lemma:Spt-MRT} and  thus omitted for brevity.
\begin{proposition} \label{Propostion:1}
	The optimal  solution to maximize \eqref{mrt_obj2} subject to \eqref{eq:P1_Simpl_C1} is  given by $x_{u}=\frac{\sqrt{P_T} \underline b_{u}}{\sqrt{\sum_{u =1}^U \underline b_{u}^2}}$, $\forall u$.
\end{proposition}

\section{Multi-ER Waveform Optimization}  \label{Sec:multi-RX WFormOPT}

In this section, we extend the single-ER power waveform optimization problem (P0) to the general case of multi-ER WPT.  In contrast to (P0), for which the phases of transmitted sinewave signals are  optimally designed to guarantee that all $N$ subtones are added constructively at one single ER, such a phase design  is generally non-optimal for a WPT system that charges  multiple ERs simultaneously, since in general $\tilde h_{1,m,u}\neq \ldots\neq \tilde h_{K,m,u}$, $\forall m,u$. Besides, to flexibly balance the harvested DC  power at different users in the multi-ER  power waveform design, the effect of ER power saturation (or diode reverse breakdown)  is crucial and thus cannot be ignored as in the single-ER case.

\subsection{Problem Formulation}

We aim to maximize the weighted sum of the harvested powers by all ERs  via jointly optimizing the subcarrier selection and amplitudes/phases of selected sub-carriers at all transmitting antennas at the ET. For convenience, we formulate the  optimization  problem in terms of the real and imaginary parts of $\mathbf{\tilde s}_m$'s separately. Specifically, let $\tilde{s}_{m,u}=\bar{s}_{m,u}+j \hat{s}_{m,u}$ and $\tilde{h}_{k,m,u}=\bar{h}_{k,m,u}+j \hat{h}_{k,m,u}$, with $s_{m,u}=\sqrt{\bar{s}_{m,u}^2+\hat{s}_{m,u}^2}$,  $h_{k,m,u}=\sqrt{\bar{h}_{k,m,u}^2+\hat{h}_{k,m,u}^2}$,  $\phi_{m,u}=\arctan(\hat{s}_{m,u}/\bar{s}_{m,u})$, and $\psi_{k,m,u}=\arctan(\hat{h}_{k,m,u}/\bar{h}_{k,m,u})$, $\forall m,u$. The received signal  at ER $k$ (given in (\ref{eq:y(t)}))  is thus re-expressed as 
\begin{align} \label{eq:y(t)_New2}
	y_k(t)= \sum_{m=1}^{M} \sum_{u =1}^{U}\sqrt{2}  \big( \bar{g}_{k,m,u}(t) \bar{s}_{m,u} + \hat{g}_{k,m,u}(t)\hat{s}_{m,u} \big),
\end{align}
where $\bar{g}_{k,m,u}(t)=\bar{h}_{k,m,u}\cos(w_u t)-\hat{h}_{k,m,u}\sin(w_u t)$ and $\hat{g}_{k,m,u}(t)=-(\bar{h}_{k,m,u}\sin(w_u t)+\hat{h}_{k,m,u}\cos(w_u t))$. Let $\Psi_{\mathrm{RHS}_k}(\{\bar{s}_{m,u}\},\{\hat{s}_{m,u}\})  $
denote the RHS of \eqref{eqyy} with $y_k(t)$ given in \eqref{eq:y(t)_New2}.  The optimization problem is thus formulated as
\begin{align} 
\mathrm{(P2)}:  
\mathop{\mathtt{maximize}}_{ \{\bar{s}_{m,u}\},\{\hat{s}_{m,u}\}} \hspace{0.2cm} \sum_{k=1}^{K} \theta_k \Psi_{\mathrm{RHS}_k}(\{\bar{s}_{m,u}\},\{\hat{s}_{m,u}\})\label{eq:P3_Obj} \hspace{-2cm}\\
\mathtt{s.t.} \hspace{4.3cm}&\hspace{-4.3cm}\sum_{m=1}^M \sum_{u=1}^U \bar{s}_{m,u}^2+\hat{s}_{m,u}^2 \le P_T, \label{eq:P3_C1} \\
\hspace{4.3cm}&\hspace{-4.3cm} {\max}\Big\{y_k\big(T(q-1)/Q\big)\Big\}_{q=1}^Q \leq  \frac{V_B}{2\sqrt{R_s}},\forall k, \label{saturation2}\\
\hspace{4.3cm}&\hspace{-4.3cm} \|\underline{\mathbf s}_m\|_0  \le N,  m =1\dots,M, \label{eq:P3_C2}
\end{align}
where $\underline{\mathbf s}_m=[(\bar{s}_{m,1}^2+\hat{s}_{m,1}^2),\dots,(\bar{s}_{m,U}^2+\hat{s}_{m,U}^2)  ]^T \in \mathbb R^{U\times 1}$ and the coefficients $\theta_k$'s  control the harvested power trade-off among different ERs, with $\sum_{k=1}^K\theta_k=1$ and $\theta_k \ge 0$, $\forall k$. The constraint in \eqref{saturation2} follows from \eqref{saturation} and it guarantees that the received power waveform at each ER does not exceed the diode reverse breakdown voltage (so as to avoid the output power saturation shown in Fig. \ref{fig:sim}). Since it is theoretically difficult to handle a continuous-time power waveform, we discretize one period of each ER's received power waveform into equally spaced $Q>0$ samples in constraint \eqref{saturation2}. It is not difficult to observe that (P2) is a non-convex optimization problem due to the non-concave objective function and the cardinality constraint in \eqref{eq:P3_C2} being  non-convex and discontinuous. Thus, solving this optimization problem is challenging  and a low-complexity algorithm for solving (P2) is highly desirable  in practice.

\textit{\underline {Remark 5.1}:} In Section \ref{amp_opt}, we optimized the power waveform for a single-ER WPT system by considering the peak amplitude of the ER received signal (cf. Frequency-MRT). However, in a multi-ER WPT system, the peak amplitude of the received signal at each ER can occur at a different time in general. Besides, the exact occurrence time of the peak amplitude of each ER’s received signal is difficult to determine, since it depends on its CSI and the  power waveform. Hence, the approach adopted  for the single-ER WPT system is difficult to be applied to solve (P2), even sub-optimally.

\subsection{DCP-based Algorithm } \label{tpm}

In view of the intractability of (P2), we propose an
iterative algorithm to solve it by utilizing the DCP technique. To deal with the non-convex cardinality constraint in \eqref{eq:P3_C2}, the Ky Fan $N$-norm \cite{fan1951maximum} is adopted to transform the cardinality constraint into the difference-of-convex-functions (DC) form. Specifically, Ky Fan $N$-norm of $\underline{\mathbf s}_m$  is given by the sum of largest-$N$ absolute values, i.e.,
\begin{align}
\vvvert\underline{\mathbf s}_m\vvvert_N &= \sum_{u=1}^{N} |\underline{ s}_{m,\pi(u)}| \nonumber\\
&= \sum_{u=1}^{N} \bar{s}_{m,\pi(u)}^2+\hat{s}_{m,\pi(u)}^2,
\end{align}
where $\pi$ is a permutation of $\{1,\dots,U\}$ and $|\underline{ s}_{m,\pi(1)}|\geq \cdots \geq |\underline{ s}_{m,\pi(U)}|$. Accordingly, if $\|\underline{\mathbf s}_m\|_0$ is no greater than $N$, $\|\underline{\mathbf s}_m\|_1$ is equal to $\vvvert\underline{\mathbf s}_m\vvvert_N$. Thus, the cardinality constraint in \eqref{eq:P3_C2} can be equivalently expressed as
\begin{align}
\|\underline{\mathbf s}_m\|_1-\vvvert\underline{\mathbf s}_m\vvvert_N =0,m=1,\dots,M. \label{dc}
\end{align} 
Since all the entries of $\underline{\mathbf s}_m$ are non-negative, $\|\underline{\mathbf s}_m\|_1$ is a composition of the non-decreasing function $\|\cdot\|_1$ and convex quadratic functions $\bar{s}_{m,u}^2+\hat{s}_{m,u}^2,\forall m,u$. Therefore, $\|\underline{\mathbf s}_m\|_1$ is convex with respect to $\bar{s}_{m,u}$'s and $\hat{s}_{m,u}$'s. Similarly, we conclude that $\vvvert\underline{\mathbf s}_m\vvvert_N$ is also convex with respect to $\bar{s}_{m,u}$'s and $\hat{s}_{m,u}$'s. As a result, constraint \eqref{dc} is a DC representation of the cardinality constraint in \eqref{eq:P3_C2}.

By replacing \eqref{eq:P3_C2} with  \eqref{dc}, (P2) is equivalently transformed to
\begin{align} 
\mathrm{(P3)}:  
\mathop{\mathtt{maximize}}_{ \{\bar{s}_{m,u}\},\{\hat{s}_{m,u}\}} & \sum_{k=1}^{K} \theta_k \Psi_{\mathrm{RHS}_k}(\{\bar{s}_{m,u}\},\{\hat{s}_{m,u}\}) \\
\mathtt{s.t.} ~\hspace{1cm}&\hspace{-1cm} \|\underline{\mathbf s}_m\|_1-\vvvert\underline{\mathbf s}_m\vvvert_N =0,  m =1\dots,M, \label{eq:P3_C22} \\
~\hspace{1cm}&\hspace{-1cm}\eqref{eq:P3_C1}, \eqref{saturation2}. \nonumber
\end{align}
It is noteworthy that the cardinality constraint of (P2) is defined by a discontinuous function, whereas the DC constraint of (P3) is defined by a continuous function, although both describe the same feasible set. However,  problem (P3) is still non-convex and difficult to be optimally solved in general due to the non-concave objective function and non-convex DC constraint in \eqref{eq:P3_C22}. To overcome such difficulties, we resort to a penalty-based method by adding DC constraint-related penalty terms to the objective function of (P3), yielding the following optimization problem
\begin{align} 
\mathrm{(P4}):  
\mathop{\mathtt{maximize}}_{ \{\bar{s}_{m,u}\},\{\hat{s}_{m,u}\}} \hspace{0.2cm} \sum_{k=1}^{K} \theta_k \Psi_{\mathrm{RHS}_k}(\{\bar{s}_{m,u}\},\{\hat{s}_{m,u}\}) \nonumber\\ &\hspace{-4.1cm} - \mu \sum_{m=1}^M \Big(\|\underline{\mathbf s}_m\|_1-\vvvert\underline{\mathbf s}_m\vvvert_N\Big) \\
\mathtt{s.t.} ~\hspace{4.7cm}&\hspace{-4.4cm}\eqref{eq:P3_C1}, \eqref{saturation2},\nonumber
\end{align}
where $\mu>0$ is the penalty parameter that imposes a cost for the constraint violation of the constraints in \eqref{eq:P3_C22}. In particular, when $\mu \rightarrow \infty$, solving the above problem yields an approximate solution to  (P3) \cite{penalty_method}. 	However, initializing $\mu$  to be a sufficiently small value generally yields a good starting point for the proposed algorithm, even though  this point may be infeasible for (P3). By gradually increasing the value of $\mu$ by a factor of $\varrho>1$, we can maximize the original objective function, i.e., $\sum_{k=1}^{K} \theta_k \Psi_{\mathrm{RHS}_k}(\{\bar{s}_{m,u}\},\{\hat{s}_{m,u}\})$, and  obtain a solution that satisfies all the equality constraints in  \eqref{eq:P3_C22} within a predefined accuracy. This thus leads to  a two-layer iterative algorithm, where the inner layer solves the penalized optimization problem (P4) while the outer layer updates the penalty coefficient $\mu$, until the convergence is achieved. 

For any given $\mu>0$, (P4) is still a non-convex optimization problem due to the DC structure of its objective function, i.e., 
\begin{align} 
\underbrace{\mu \hspace{-1mm} \sum_{m=1}^M \hspace{-1mm} \vvvert\underline{\mathbf s}_m\vvvert_N +  \hspace{-1mm} \sum_{k=1}^{K} \theta_k \Psi_{\mathrm{RHS}_k}(\{\bar{s}_{m,u}\},\{\hat{s}_{m,u}\})}_{\text{convex}:f_1(\cdot)}   - \underbrace{\mu \hspace{-1mm}\sum_{m=1}^M \|\underline{\mathbf s}_m\|_1}_{\text{convex}:f_2(\cdot)}, \nonumber
\end{align}
for which we can apply the concave-convex procedure (CCCP) \cite{mm_tsp} to approximately solve it in an iterative manner.  Specifically, at each iteration $l=1,2,\dots,$ we approximate $f_1$ by a linear function using its first-order Taylor series  to form a convex approximate optimization problem.  Let $\bar{s}_{m,u}^{(l)}$ and $\hat{s}_{m,u}^{(l)}$, $\forall m,u$, denote the  values of decision variables at the beginning of iteration $l$. The linear approximation of $f_1$, denoted by $\underline {f_1}$, is given by (with constant terms ignored)
\begin{multline} 
\underline {f_1}(\bar{s}_{m,u},\hat{s}_{m,u}) = \mu\sum_{m=1}^M \sum_{u=1}^U \bar g_{m,u}^{(l)}\bar{s}_{m,u} + \hat g_{m,u}^{(l)}\hat{s}_{m,u}  \\ +\sum_{k =1}^K \sum_{m=1}^M \sum_{u=1}^U \theta_k  \Big(\bar{\beta}_{k,m,u}^{(l)}\bar{s}_{m,u}   + \hat{\beta}_{k,m,u}^{(l)}\hat{s}_{m,u} \Big), 
\end{multline}
with the constant coefficients given by 
\begin{align} 
\big\{\bar g_{m,u}^{(l)},\hat g_{m,u}^{(l)}\big\} &= 
\begin{cases} 
\big\{2\bar{s}_{m,u}^{(l)},2\hat{s}_{m,u}^{(l)}\big\} & \text{if  } ~~\pi(i)\leq N, \\
0      & \text{if  } ~~\pi(i)> N,
\end{cases} \\
\bar{\beta}_{k,m,u}^{(l)}&=\dfrac{1}{T} \int_{T} \dfrac{\sqrt{2R_s}}{\eta V_0} \bar{g}_{k,m,u}(t) z_k^{(l)}(t)  dt, \label{n2} \\
\hat{\beta}_{k,m,u}^{(l)}&=\dfrac{1}{T} \int_{T} \dfrac{\sqrt{2R_s}}{\eta V_0} \hat{g}_{k,m,u}(t) z_k^{(l)}(t)  dt, \label{n3}
\end{align}
where $z_k^{(l)}(t) \triangleq \exp((\sqrt{2R_{s}}  \sum_{m=1}^M \sum_{u =1}^U \bar{g}_{k,m,u}(t) \bar{s}_{m,u}^{(l)}+ \hat{g}_{k,m,u}(t)\hat{s}_{m,u}^{(l)})/(\eta V_0))$. Accordingly,   the approximate problem of (P4) for each iteration $l$ is given by 
\begin{align} 
\mathrm{(P4}-l):  
\mathop{\mathtt{maximize}}_{ \{\bar{s}_{m,u}\},\{\hat{s}_{m,u}\}}~ ~& \underline {f_1}(\bar{s}_{m,u},\hat{s}_{m,u}) -\mu \hspace{-1mm}\sum_{m=1}^M \|\underline{\mathbf s}_m\|_1 \nonumber \\
\mathtt{s.t.} ~\hspace{0.2cm}&\hspace{-0.2cm}~~\eqref{eq:P3_C1}, \eqref{saturation2}.\nonumber
\end{align}
Since the objective function of (P4$-l$) is concave,  (P4$-l$) is a convex  quadratically constrained quadratic  programming (QCQP) problem, where it can be easily solved by using numerical convex program solvers such as CVX \cite{227}. The overall DCP-based algorithm
to solve (P2) is summarized  in Algorithm 1.

\begin{algorithm} [t!]
	\caption{DCP-based Algorithm for Solving (P2) }\label{Tabel_3}
	\begin{algorithmic}[1]
		\STATE \textbf{Initialize:} { $\{\bar{s}_{m,u}\}$, $\{\hat{s}_{m,u}\}$, and $\mu>0$ }
		\REPEAT
		\REPEAT
		\STATE Update $\{\bar{s}_{m,u}\}$ and $\{\hat{s}_{m,u}\}$ by solving (P4$-l$).
		\UNTIL{{The fractional decrease of the objective value of (P4$-l$) is below a prescribed threshold $  \epsilon_1>0$ or a given  maximum number of iterations is reached.}	}
		\STATE Update the penalty coefficient as $\mu \leftarrow \varrho\mu$.
		\UNTIL{{The constraint violation $ \big(\sum_{m=1}^M\|\underline{\mathbf s}_m\|_1-\vvvert\underline{\mathbf s}_m\vvvert_N\big)$  is below a prescribed threshold $  \epsilon_2>0$.}}
	\end{algorithmic}
\end{algorithm}

\section{Simulation Results}

\subsection{Performance Evaluation for Single-ER WPT} \label{sec3:simulationl}

We consider a single-user WPT system, where  an  ET with $M=8$ antennas delivers wireless power to an ER with a single antenna. 
The central frequency  and the total available  bandwidth are given as $f_c=915$ MHz and $B=10$ MHz, respectively. Moreover, we assume that $U=160$ with $\Delta_u = 62.5$ KHz. For the channel from the ET to ER, we assume $45.65$ dB path loss  and a NLoS channel power delay profile with $5$ distinct paths. For simplicity, it is assumed that the signal power is equally divided among all paths. However, the delay of each path and its phase are assumed to be  uniformly distributed over $[0,~0.3~\mu \text{s}]$  and $[0,~2\pi]$, respectively. Fig. \ref{fig:channels} illustrates one realization of the frequency response of the assumed channel, which will be used in the simulations in this subsection. The ohmic resistance of the ER's antenna and load are  set as $R_s=50~\Omega$ and $R_L=10$ k$\Omega$, respectively. The diode parameters  are given by $I_0=3~\mu$A, $I_{BV}=300~\mu$A, $V_0=25.86$ mV, $V_{B}=3.8$ V, and $\eta=1.05$ (in accordance with HSMS-285x datasheet \cite{datasheet}).

\begin{figure}[t!]
	\centering
	\includegraphics[trim = 0mm 0mm 0mm 0mm, clip,width=0.8\linewidth]{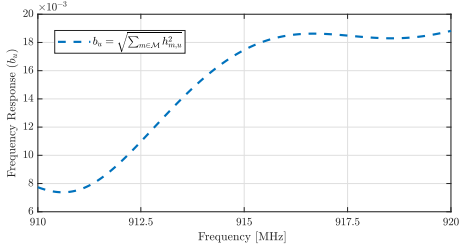}
	\caption{Frequency response of the multipath channel.}  
	\label{fig:channels} \vspace{-0.2cm}
\end{figure} 

\begin{figure}[t!]
	\centering
	\includegraphics[width=0.8\linewidth]{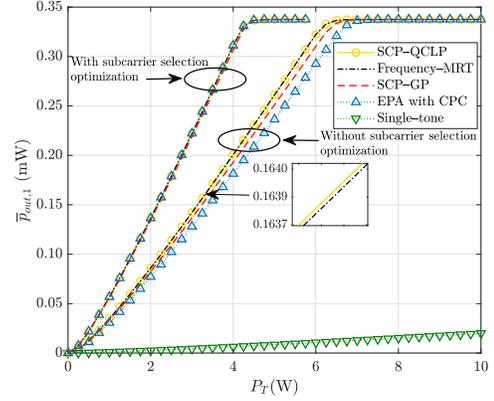}
	\caption{Output DC power versus maximum transmit power when $ N=32. $}  
	\label{fig:power} \vspace{-0.5cm}
\end{figure}

For comparison with our proposed subcarrier selection  method in Proposition \ref{sub_allocation_pro}, we consider $U=N$ equally-spaced subcarriers without  selection over the total available bandwidth as a benchmark design. Moreover, for comparison with our proposed frequency-MRT solution  and SCP-QCLP algorithm,  we consider  three benchmark designs: i)  the SCP-GP algorithm \cite{Clerckx:2016b} with the $4$-th order  truncated Taylor approximation model of the diode;   ii) equal power allocation (EPA) over all ET antennas and frequency tones with channel phase compensation (CPC), i.e., $\phi_{m,u}=-\psi_{m,u}$,  $\forall m,u$;  and iii) single-tone power allocation \cite{888} based on the linear model of the diode.  
To implement the SCP-GP algorithm,  we set its stopping criteria as $\epsilon=10^{-3}$ \cite{Clerckx:2016b}. We also consider the EPA solution as the initial point for both SCP-QCLP and SCP-GP algorithms.

Fig. \ref{fig:power} and Fig. \ref{fig:tones} show the output DC power at the ER by varying  $P_T$ with fixed $N=32$, and by varying  $N$ with fixed $P_T=10$ W, respectively, under different waveform design schemes. It is observed that SCP-QCLP achieves the best performance over all values of  $P_T$ and $ N $, while the frequency-MRT solution performs very close to it. This suggests that the frequency-MRT solution provides a practically appealing alternative solution for the  single-ER WPT system considering its low complexity.

\begin{figure}[t!]
	\centering
	\includegraphics[trim = 0mm 0mm 0mm 0mm, clip,width=0.8\linewidth]{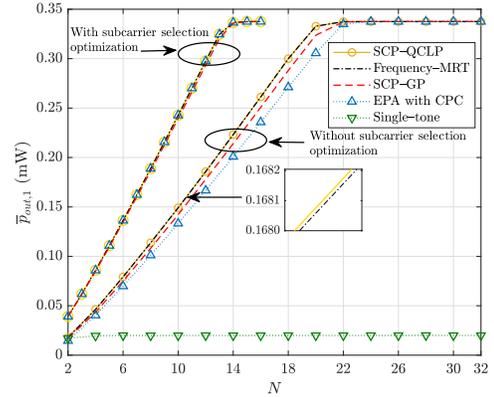}
	\caption{Output DC power versus  number of frequency tones  when $ P_T = 10 $ W.}  
	\label{fig:tones} \vspace{-0.5cm}
\end{figure}

It is also observed that when $ P_T $ and $ N $ are smaller, SCP-GP performs close to SCP-QCLP. However, the gap between them  increases with   $P_T$ or $ N $,  explained as follows. In \cite{Clerckx:2016b}, SCP-GP is proposed based on the truncated Taylor approximation model of the diode, which is valid when the voltage drop across the diode, i.e., $v_d(t)$ shown in Fig. \ref{fig:first}, is sufficiently small. However, by increasing the transmit power $P_T$, the peak voltage collected by the ER  increases, which causes  the voltage drop across the diode to increase.  Thus, the accuracy of the truncated Taylor approximation model of the diode reduces, and the performance  of SCP-GP degrades. Moreover,  SCP-QCLP converges much faster than SCP-GP. This is due to the fact that SCP-QCLP requires   to solve a simple QCLP problem at each iteration only, while SCP-GP needs to solve a GP problem at each iteration, which requires more computational time. The convergence time of SCP-QCLP and SCP-GP is  compared in Table \ref{tab:perf_compar}.\footnote{Simulations are implemented on MATLAB R2017b and tested on a PC with a Core i7-6700 CPU, 8-GB of RAM, and Windows 7.}
\begin{table}[h]
	\centering
	\caption{Convergence time comparison.} \vspace{-0mm}
	\label{tab:perf_compar}
	\begin{tabular}{|c|S|S|S|}
		\hline
		\multirow{2}{*}{Design  approach} & \multicolumn{3}{c|}{Convergence time (second)} \\ \cline{2-4} 
		& $N=4$ & $N=16$ & $N=32$ \\ \hline
		SCP-QCLP & 0.032 & 0.052 & 0.220 \\ \hline
		SCP-GP~~~~~& 6.325 & ~22.742 & 49.465 \\ \hline
	\end{tabular} 
\end{table}

\begin{figure} 
	\centering
	\begin{subfigure}[b]{0.4\textwidth} %0.242
		\centering
		\includegraphics[trim = 0mm 0mm 0mm 0mm, clip,width=\columnwidth]{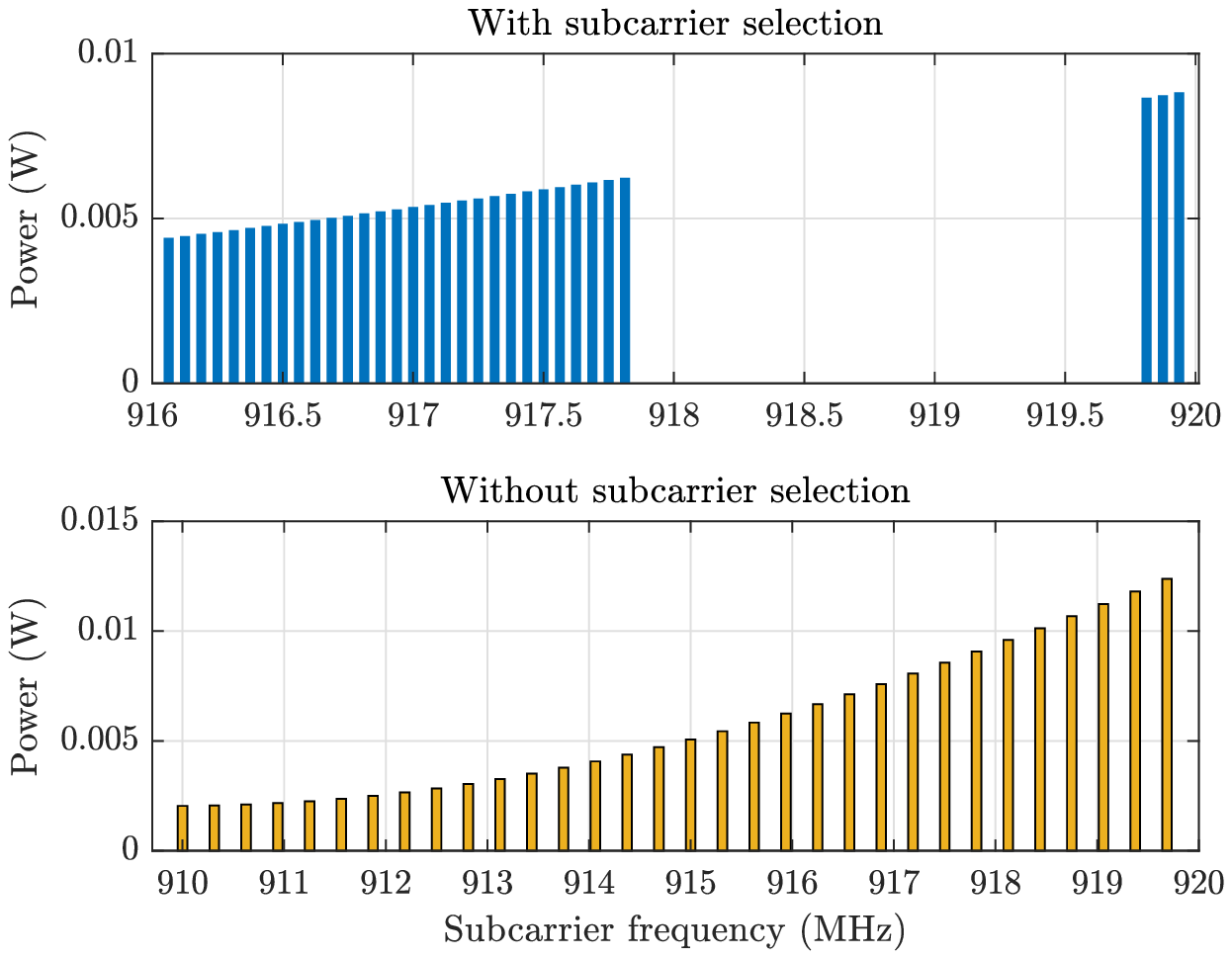}
		\caption{Frequency-domain power allocation.}
		\label{fig_3_a} 
	\end{subfigure}
	\hspace{-0.1cm}
	\begin{subfigure}[b]{0.4\textwidth} 
		\centering
		\includegraphics[trim = 0mm 0mm 0mm 0mm, clip,width=\columnwidth]{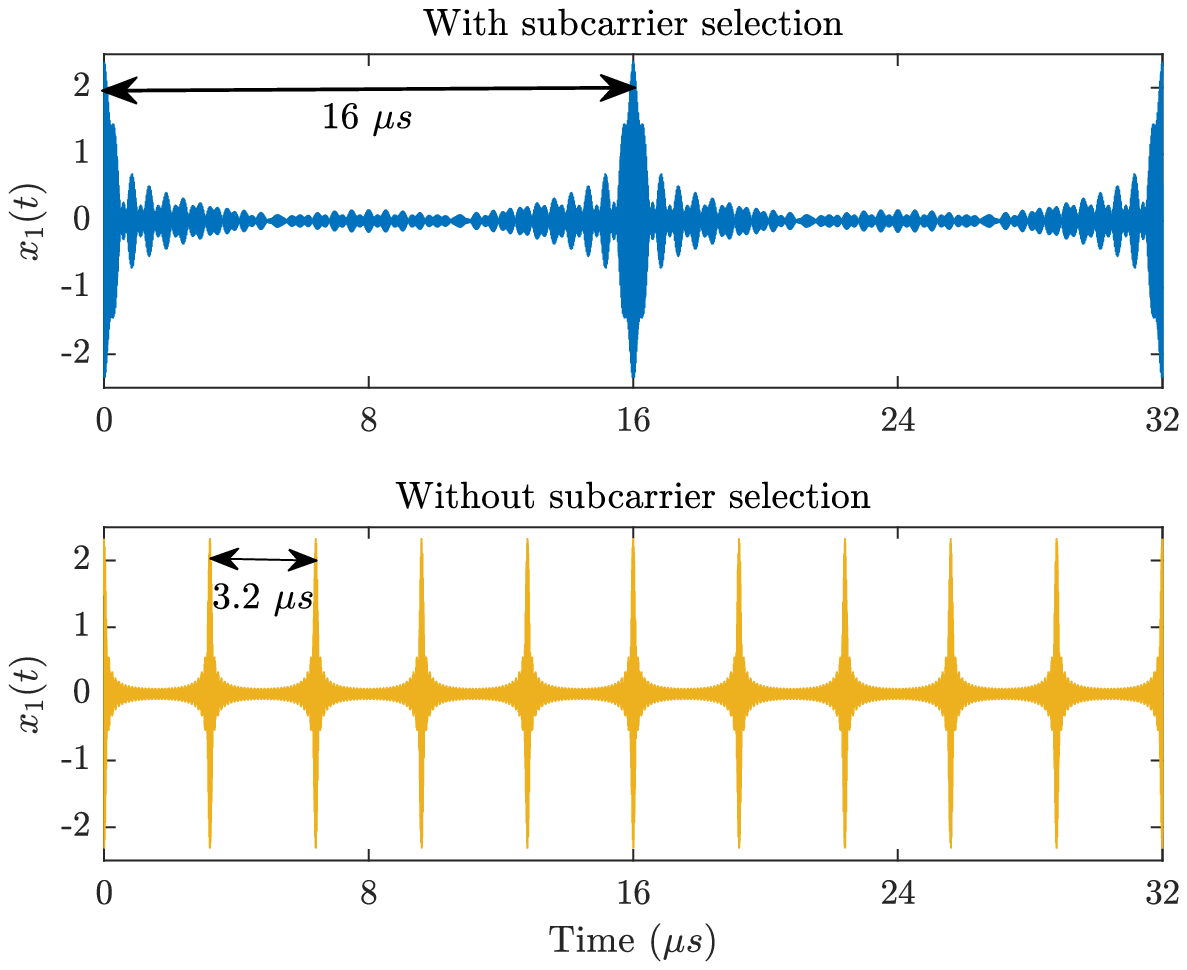}
		\caption{Time-domain  waveform. }
		\label{fig_3_b} 
	\end{subfigure}
	\caption{Frequency-domain power allocation and time-domain waveform of the optimized power signal on the first ET antenna,   when $N=32$ and $P_T=4$ W.} \label{reflection} 
	\label{wav} \vspace{-0.5cm}
\end{figure} 

It is further  observed from Fig. \ref{fig:power} and Fig. \ref{fig:tones} that subcarrier selection optimization leads to significantly    better performance while reaching to the power saturation with smaller $P_T$ and $ N $, as compared to the conventional design with equally-spaced sub-carriers (even after the same  power waveform optimization). Moreover, the EPA scheme with optimized  subcarrier selection  performs even better than SCP-QCLP without subcarrier selection. 	Thus, given a certain   bandwidth in practice for WPT, it is not optimal to firstly assign a fixed number ($U=N$) of equally-spaced subcarriers and then optimize the power waveform based on them; instead, our results show that it is more efficient to assign a larger number ($U$ with $U>N$) of equally-spaced subcarriers with smaller inter-subcarrier spacing,  then based on the CSI select $N$ out of $U$ subcarriers and further optimize the corresponding power waveform.

To further illustrate the effect of subcarrier selection on the optimized power signals, we show in Fig. \ref{wav}(a) and Fig. \ref{wav}(b) their frequency-domain power allocations and time-domain waveforms on the first ET antenna, respectively, for the two cases with subcarrier selection ($U=160$ and $N=32$) versus without subcarrier selection ($U=N=32$), both using SCP-QCLP with $P_T=4$ W.
It is observed from Fig. \ref{wav}(a) that the subcarrier selection results in  reduced and unequal inter-subcarrier spacing. In addition, the optimized power waveforms shown in Fig. \ref{wav}(b) have the PAPRs of 17.9 dB and 17.7 dB for the cases with subcarrier selection and without subcarrier selection, respectively, which are similar. This is expected since for both cases, the same number of $N=32$ subcarriers are used and the PAPR of the power waveform is mainly determined by $N$.  It is further observed from Fig. \ref{wav}(b) that the period $T$ of the optimized power signal with  subcarrier selection  is larger than that without subcarrier selection, since $T=1/\Delta_u$ and $\Delta_u$ in the former case is much smaller than that in the latter case. Note that the diode in the ER rectifier is turned on less frequently  as $T$ becomes larger  (see Fig. \ref{fig3:first}). Nevertheless, this  does not affect the ER harvested power with subcarrier selection (or larger $T$)  as long as a capacitor with sufficiently  high Q factor is used, such that the signal peak will not decay much over one period of the power  signal.

%However, since the subcarrier selection results in  reduced inter-subcarrier spacing, the period $T$ of the optimized power signal with the subcarrier selection  is larger than that without subcarrier selection, since $T=1/\Delta_u$ and $\Delta_u$ in the former case is much smaller than that in the latter case (see Fig. \ref{wav}(a)). Note that the diode in the ER rectifier is turned on less frequently  as $T$ becomes larger  (see Fig. \ref{fig3:first}). Nevertheless, this  does not affect the ER harvested power with subcarrier selection (or larger $T$)  as long as we use a large capacitor with high Q factor, so the signal peak will not decay much over one period of the power  signal. 

\subsection{Performance Evaluation for Multi-ER WPT} \label{sec4:simulationl_multi} 

\begin{figure}[t!] 
	\centering
	\includegraphics[trim = 0mm 0mm 0mm 0mm, clip,width=0.7\linewidth]{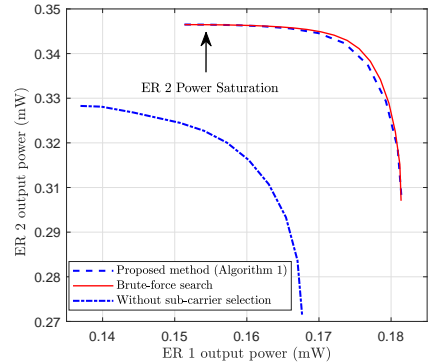} 
	\caption{Power region of harvested DC power  with $ K=2 $, $ U=16 $, and $ N=8 $.}  
	\label{fig:boundary} \vspace{-0.5cm}
\end{figure}

Next, we  consider a multi-user WPT system, where a single ET with $ M = 4 $ antennas delivers wireless power to two ERs  that are located with $6$ meters (m) and $4$ m distances from the ET, respectively, each with a single antenna. Other system/channel  parameters are the same as in Section \ref{sec3:simulationl} (if not specified otherwise).

First,  we compare the performance in terms of the Pareto boundary of the harvested DC power region  with different  user weight pairs $ (\theta_1,\theta_2) $, among our proposed solution
(Algorithm 1), and Algorithm 1 without sub-carrier selection (i.e., with $U=N$ equally-spaced sub-carriers), as well as a brute-force search with sub-carrier selection. For  the brute-force  search, we randomly generated $10^6$ number of power waveforms (subject to the constraints in \eqref{eq:P3_C1}-\eqref{eq:P3_C2}) for each user weight pair and selected the power waveform that results in the maximum weighted  sum of the output powers.  We consider $ U=16 $, $ N=8 $, and $ P_T = 5 $ W. It is observed from Fig. \ref{fig:boundary} that our proposed method performs very close to the brute-force search, and also significantly outperforms the scheme without sub-carrier selection.  

\begin{figure}[t!] 
	\centering
	\includegraphics[trim = 0mm 0mm 0mm 0mm, clip,width=0.7\linewidth]{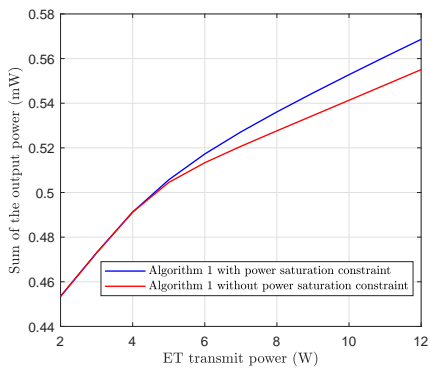} 
	\caption{Sum output power versus ET transmit power with $ K=2 $, $ U=16 $, and $ N=8 $.}  
	\label{fig:boundary2} \vspace{-0.5cm}
\end{figure}

Next, we plot in Fig. \ref{fig:boundary2} the sum of the output power of ER 1 and ER 2 by varying $P_T$ with $\theta_1=\theta_2=0.5$.  In particular, we compare the performance of Algorithm 1 with  the power saturation constraint in  \eqref{saturation2} considered or  not. It is observed that adopting the power saturation constraint results in better performance with increasing $P_T$,  which is explained as follows. When $P_T$ increases, ER 2 (nearer to the ET) first  reaches the power saturation. In this case, if the power saturation constraint is considered for the subcarrier selection and power waveform design, the ET will not further increase ER 2's output power, but instead  aim to enhance ER 1's output power in its power waveform design. The above result shows that the consideration of the ER  power saturation effect captured in our newly derived non-linear rectifier model is crucial to the power waveform optimization in the multi-ER WPT system, especially when the ERs' channels are heterogeneous  in practice.   

\section{Conclusion}

In this paper, we studied the power waveform design problem for a multi-antenna  multi-user WPT system under  the nonlinear rectenna  model. First, we developed a refined rectifier  model based on circuit analysis which accurately captures not only  the nonlinearity of the rectifier circuit  with respect to the power  waveform, but also its intrinsic power saturation effect. Next, based on this new model and by adopting a multisine power waveform structure, we formulated new problems to jointly optimize the subcarrier selection and power waveform to maximize the harvested DC power at the ER(s) for both the cases of single-ER and multi-ER WPT. We proposed efficient algorithms for solving these challenging optimization problems by leveraging non-convex optimization techniques such as  SCP, CCCP, and DCP. The new rectenna  model was validated by circuit simulation, and the performance of the proposed algorithms was shown to be superior to the existing/benchmarking schemes based on the conventional rectenna model. In particular, it was shown that the subcarrier selection and consideration of rectenna power saturation lead to significantly better performance (in terms of the output DC power) of designed  power waveform  for both single-ER and multi-ER WPT systems and multi-ER WPT system, respectively.

\bibliographystyle{IEEEtran}  
\footnotesize{\bibliography{bibfile}}
\end{document}